\documentclass[letterpaper, 10 pt, conference]{ieeeconf}
\IEEEoverridecommandlockouts
\usepackage{cite}
\usepackage{amsmath,amssymb,amsfonts}
\usepackage{graphicx}
\usepackage{textcomp}
\usepackage{amsmath}
\usepackage{bm}
\usepackage{subcaption}
\usepackage{amsmath}
\usepackage{algorithm}
\usepackage{algpseudocode}
\usepackage{xcolor}
\usepackage{booktabs}
\usepackage{array}
\newtheorem{assumption}{Assumption}
\newtheorem{remark}{Remark}
\newtheorem{lemma}{Lemma}

\newtheorem{theorem}{Theorem}[section]

\def\BibTeX{{\rm B\kern-.05em{\sc i\kern-.025em b}\kern-.08em
    T\kern-.1667em\lower.7ex\hbox{E}\kern-.125emX}}

\begin{document}

\title{Distributed Invariant Unscented Kalman Filter based on Inverse Covariance Intersection with Intermittent Measurements
}

\author{%
Zhian Ruan$^{*}$, Yizhi Zhou$^{*}$
 \thanks{Y. Zhou is with the Electrical and Computer Engineering, George Mason University, Fairfax, VA, USA. Z. Ruan is with the Mechanical Engineering, Northwestern University, Evanston, IL, 
 USA. (email {\tt\small zhianruan2025@u.northwestern.edu, yzhou26@gmu.edu}). 
Corresponding author: Yizhi Zhou%
\thanks{$^{*}$ The authors contributed equally as co-first authors.}
}
}

\maketitle

\begin{abstract}
This paper studies the problem of distributed state estimation (DSE) over sensor networks on matrix Lie groups, which is crucial for applications where system states evolve on Lie groups rather than vector spaces. We propose a diffusion-based distributed invariant Unscented Kalman Filter using the inverse covariance intersection (DIUKF-ICI) method to address target tracking in 3D environments. Unlike existing distributed UKFs confined to vector spaces, our approach extends the distributed UKF framework to Lie groups, enabling local estimates to be fused with intermediate information from neighboring agents on Lie groups. To handle the unknown correlations across local estimates, we extend the ICI fusion strategy to matrix Lie groups for the first time and integrate it into the diffusion algorithm. We demonstrate that the estimation error of the proposed method is bounded. Additionally, the algorithm is fully distributed, robust against intermittent measurements, and adaptable to time-varying communication topologies. The effectiveness of the proposed method is validated through extensive Monte-Carlo simulations.
\end{abstract}

\section{Introduction}
State estimation over sensor networks has emerged as a pivotal research area with applications spanning diverse sectors such as Unmanned Aerial Vehicles (UAVs), robotic manipulators, and surveillance systems \cite{KUMAR2021109558}. Target tracking within sensor networks is a fundamental problem in state estimation, which can be addressed using either centralized or distributed methods. Compared to centralized approaches, which are vulnerable to single node failures and require extensive communication resources, distributed state estimation (DSE) has garnered increased attention for its scalability, robustness, and efficiency. In distributed algorithms, each sensor node provides a local estimate and fuse its neighbors' information to compute an improved estimate. However, because the sensors typically use common system models and share environmental data, the local estimates are often mutually correlated, presenting a challenge for accurate fusion \cite{NOACK2017}.

To address DSE problems while managing unknown correlations, algorithms based on Extended Kalman Filters (EKFs) and Unscented Kalman Filters (UKFs) are two commonly used approaches. Distributed Kalman Filters have been proposed using either the diffusion method with covariance intersection (CI) \cite{Hu2012} or the consensus method \cite{Ren2007}. Compared to the consensus method, which requires multiple iterations of communication and information exchange between agents to reach agreement in a single timestep, the diffusion method directly fuse the information from one-hop communication to update the local estimate, making it more practical in real-world applications.
To develop more practical DSE algorithms for general nonlinear systems and avoid the linearization required by EKF, Hao et al. \cite{CHEN2021} introduced a CI-based distributed diffusion UKF with intermittent measurements, which eliminates the need for linearization. To reduce the conservativeness of CI-based algorithms, inverse covariance intersection (ICI) \cite{NOACK2017} has been introduced to address DSE challenges \cite{Sun2023}. However, it is important to note that all these methods are based on \textit{vector spaces} with additive errors, limiting their application to 2D systems or simple linear models. In 3-D environments, a system state generally contains a 3-D rotation represented by a quaternion \cite{ZYICI} or a rotation matrix, which doesn’t belong to the vector space and makes the system highly nonlinear. While Euler angles can represent 3D rotation, they suffer from the well-known Gimbal lock problem.

To address this issue, recent works have introduced Lie groups for state representation, particularly $SE(3)$ for robot pose \cite{Jie2023}, and adapted DSE algorithms from vector spaces to the context of Lie groups \cite{ZH2024}. Split CI approach on Lie groups is presented in \cite{Liang2021} to fuse multiple correlated poses. Xu et al \cite{Jie2023} propose a fully distributed invariant Extended Kalman Filter (IEKF) which extends the CI algorithm to matrix Lie groups while ensuring the estimator's consistency. Lee et al \cite{LE2023} propose a maximum likelihood estimation approach on Lie groups for mobile network self-tracking.
Due to the conservativeness of the CI-based methods which directly ignore the cross-correlations, refs \cite{Zarei2024} derives an optimization-based approach to iteratively estimate the cross-correlation terms, integrating it in the EKF framework.  IEKF and its related fusion algorithms demonstrate good performance, characterized by a larger convergence domain. Nevertheless, an important drawback of the IEKF-based method is that it requires linearizing the system for yielding an invariant error dynamics, which might be difficult to carry out for complex nonlinear systems. 

This paper proposes a novel distributed invariant UKF based on ICI with intermittent measurements over sensor networks, which is fully distributed. The main contributions of our work are as follows: 1) To the best of our knowledge, most existing distributed UKFs operate only in vector spaces. In contrast, this paper extends UKF-based DSE methods to matrix Lie groups, allowing for state estimation and fusion in systems that evolve on Lie groups. 2) To properly handle the unknown cross-correlations between local estimates on Lie groups, we extend the ICI algorithm from vector spaces to Lie groups for the first time and integrate it into the diffusion-based framework. We also prove that the fused estimates for each agent are bounded in this work, a result not previously provided in the literature for solving DSE problems on Lie groups. 3) The performance of the proposed method is validated through extensive Monte-Carlo simulations and compared to the CI-based fusion method, demonstrating superior accuracy.

\section{Preliminaries and Problem Statement}

\subsection{Matrix Lie Group}

A matrix Lie group $\mathcal{G}$ \cite{BAB2017} is a subset of square invertible matrices with the following properties hold
\begin{align}
    \forall a \in \mathcal{G}, a^{-1} &\in \mathcal{G} \\
    \forall a, b \in \mathcal{G}, ab &\in \mathcal{G}
\end{align}
and its associated Lie Algebra, denoted as $\mathfrak g$, represents the tangent space of $\mathcal{G}$. An element $\bm\xi\in \mathbb R^{\dim\mathfrak g}$ can be mapped to its corresponding Lie group using the exponential map, $\exp: \mathbb R^{dim \mathfrak g}$, given as
\begin{align}
\exp(\bm\xi)=\exp_m(\bm\xi^\wedge)
\end{align}
where $\exp_m(\cdot)$ represents the matrix exponential, and $(\cdot)^\wedge:\mathbb{R}^{\dim \mathfrak{g}} \rightarrow \mathfrak{g}$ is the linear map that transfers the element in the Lie algebra to the corresponding matrix representation. The logarithm map, which is the inverse function of the exponential map, is denoted by $\log(\cdot)$, leading to
\begin{align}
\log(\exp(\bm\xi))=\bm\xi
\end{align}
Let $\mathbf X^t\in\mathcal G$ be a state of a dynamic system at time $t$, and $\hat{\mathbf X}^t$ denotes the state estimate.The right and left invariant error between $\mathbf X^t$ and $\hat {\mathbf X}^t$, denoted as $\bm\eta_r^t$ and $\bm\eta_l^t$, respectively, are defined by
\begin{align}\label{eq_right_inv_er}
\bm\eta_r^t=\hat{\mathbf X}^t({\mathbf X}^t)^{-1}, \, \bm\eta_l^t=({\mathbf X}^t)^{-1} \hat{\mathbf X}^t
\end{align}
\begin{remark}
The derivation in this paper is based on the right-invariant error, and for detailed properties of invariant errors, we refer the reader to \cite{BAB2017}.
\end{remark}

\subsection{Unscented Kalman Filter on Lie Group}
Unlike the standard UKF or its square-root implementation, which operates in vector space, the invariant UKF \cite{BMB2017} applies the unscented transform on Lie groups and uses Lie exponential coordinates to derive uncertainty ellipsoids, as outlined in Algorithm \ref{alg2}. This forms the foundation for deriving the proposed DIUKF-ICI in this paper.

\subsection{System Model and Problem Statement}
Consider a sensor network of $N$ sensor nodes where each sensor possesses the capability of sensing and communication. The state space model associated with the environment and the measurement of each agent can be described by
\begin{align}\label{eq_sys}
\mathbf X^{t+1}&=f(\mathbf X^t, \mathbf u^t, \mathbf n^t)\nonumber\\
\mathbf z_i^t&=h_i(\mathbf X^t, \mathbf w_i^t)
\end{align}
where $\mathbf X^t\in\mathcal G$ is the state of interest at time $t$, and $\mathbf z_i^t$ is the measurement of the $i$'s sensor node. $\mathbf n^t$ and $\mathbf w_i^t$ are process noise and measurement noise, respectively, which are assumed to be white Gaussian noises and mutually uncorrelated. The covariance matrices of $\mathbf n^t$ and $\mathbf w_i^t$ are given by $\mathbf O^t$ and $\mathbf Q_i^t$. $f(\cdot):\mathcal G \rightarrow\mathcal G$ denotes the nonlinear state transition function and $h_i(\cdot)$ describes the nonlinear measurement model of the $i$'s sensor.
The topology of the sensor network of all agents is modeled by a time-varying undirected graph $\mathbb{G}^t = (\mathcal{V}, \mathcal{E}^t)$, where $\mathcal{V}$ represents the set of all agents, and $\mathcal{E}^t$ stands for the set of communication links at time $t$ defined as $\mathcal{E}^t\in\mathcal{V}\times\mathcal{V}$. In particular, node $j$ is the neighbor of agent $i$ and can communicate with the $i$'s node when $(j,i)\in\mathcal{E}^t$. We define the set of node $i$'s communication neighbors as 
$\mathcal{N}_i^t = \{j \mid (j, i) \in \mathcal{E}^t, \, j \in \mathcal{V}\}$. We assume that self communication always exists, i.e., $(i, i) \in \mathcal{E}^t, \forall i \in \mathcal{V}$.

Given the system model specified in \eqref{eq_sys}, the major problem is for every agent $i$ in the sensor network to compute a stable estimate of the unknown state $\mathbf X^t$ on the Lie group with a time-varying communication topology, while only sharing information with its neighbors. Unlike the DSE problem on the vector space which already been well established, no conclusions have been made on how to guarantee the stability estimate on Lie group for a time-varying system. 



\section{Distributed Invariant Unscented Kalman Filter Based on ICI}
In this section, we propose a diffusion-based distributed Unscented Kalman Filter on the Lie group using the ICI algorithm (DIUKF-ICI), as outlined in Algorithm \ref{alg1}. 

\begin{algorithm}[h]
\caption{DIUKF-ICI}
\label{alg1}
\begin{algorithmic}[1]
\State \textbf{Step 1: Initialization}
\State \quad Consider the nonlinear state-space model (1). Start with $\mathbf{\hat X}_i^0 = \mathbb{E}(\mathbf X_i^0)$, $\mathbf{\hat P}_i^0 = \mathbf{P}_i^0$, for all nodes $i=1,...,N$.
\State \textbf{Step 2: Individual estimation by the local filter}
\State \quad For each agent $i=1,...,N$, performs the individual estimation to compute the individual estimate $(\mathbf{\bar X}_i^t, \mathbf{\bar P}_i^t)$ by its local UKF (\textbf{cf.} Algorithm \ref{alg2}).
\State \textbf{Step 3: Incremental update}
\State \quad Each agent updates its individual estimate with the neighbor's measurements to compute the local estimate $(\mathbf{\check X}_i^t, \mathbf{\check P}_i^t)$, by \eqref{eq_inp1} and \eqref{eq_inp2}.
\State \textbf{Step 4: Diffusion update with ICI}
\State \quad Each agent fuses the local estimate with its neighbor to compute an improved fused estimate $(\mathbf{\hat X}_i^t, \mathbf{\hat P}_i^t)$
\end{algorithmic}
\end{algorithm}

\begin{algorithm}[h]
\caption{UKF on Lie group for each agent}
\label{alg2}
\begin{algorithmic}[1]
\State \textbf{Input:} $\hat{\mathbf X}_i^{t-1}, \hat{\mathbf P}_i^{t-1}, \mathbf u_i^{t-1}, \mathbf O^{t-1}, \mathbf Q_i^{t-1}$, $\mathbf{z}_i^t$
\State \textbf{Step 0: UKF parameter initialization}
\State \quad Initialize the UKF parameters: $\gamma_i, \gamma_i^\prime$, $W_{i,k}$, $W_{i,k}^\prime$,$L$, and $L^\prime$
\State \textbf{Step 1: Propagation}
\State \quad ${\mathcal P}_i^{t-1}=\text{blkdiag}(\hat{\mathbf P}_i^{t-1}, \mathbf O^{t-1})$ \hfill \textcolor{blue}{// augment covariance }
\State \quad $\bar{\mathbf X}_i^{t|t-1}=f(\hat{\mathbf X}_i^{t-1},\mathbf u_i^{t-1})$\hfill \textcolor{blue}{// propagate the noiseless mean}
\State \quad \textcolor{blue}{// generate sigma point}
\State \quad $\bm{\alpha}_{i,k}=\text{col}(\sqrt{\gamma_i{\mathcal P}_i^{t-1}})_k,\, k=1,2,...,L$ 
\State \quad $\bm{\alpha}_{i,k}=-\text{col}(\sqrt{\gamma_i{\mathcal P}_i^{t-1}})_k,\, k=L+1,L+2,...,2L$ 
\State \quad $\begin{bmatrix}{\bm\zeta_{i,k}}& \mathbf n_k\end{bmatrix}=\bm{\alpha}_{i,k},\, k=1,2,...,2L$
\State \quad \textcolor{blue}{// sigma point propagation}
\State \quad ${\mathcal X}_{i,k}^t=f(\exp({\bm\zeta_{i,k}})\hat{\mathbf X}_i^{t-1},\mathbf u_i^{t-1}),\, k=1,2,...,2L$
\State \quad ${\bm\zeta_{i,k}}\leftarrow \log({\mathcal X}_{i,k}^t(\bar{\mathbf X}_i^{t|t-1})^{-1})$
\State \quad \textcolor{blue}{// compute the propagated covariance}
\State \quad $\bar{\mathbf P}_i^{t|t-1}=\sum_{k=1}^{2L}W_{i,k} \log(\bm\zeta_{i,k}) \log(\bm\zeta_{i,k})^\top$
\State \textbf{Step 2: Update} 
\State \quad ${{\mathcal P}_i^{t}}^\prime=\text{blkdiag}(\bar{\mathbf P}_i^{t|t-1}, \mathbf Q_i^{t})$ \hfill \textcolor{blue}{// augment covariance }
\State \quad $\bar{\mathbf z}_{i,0}^t=h(\bar{\mathbf X}_i^{t|t-1},\bm 0)$\hfill \textcolor{blue}{// propagate the noiseless mean by observation model}
\State \quad \textcolor{blue}{// generate sigma point}
\State \quad $\bm{\alpha}_{i,k}^\prime=\text{col}(\sqrt{{\gamma_i^\prime}{\mathcal P}_i^{t}})_k,\, k=1,2,...,L^\prime$
\State \quad $\bm{\alpha}_{i,k}^\prime=-\text{col}(\sqrt{\gamma_i^\prime{\mathcal P}_i^{t}})_k,\, k=L^\prime+1,L^\prime+2,...,2L^\prime$ 
\State \quad $\begin{bmatrix}{\bm\zeta_{i,k}^\prime}& {\mathbf w_{i,k}^\prime}\end{bmatrix}=\bm{\alpha}_{i,k}^\prime,\, k=1,2,...,2L^\prime$
\State \quad \textcolor{blue}{// sigma point propagation by observation model}
\State \quad $\bar{\mathbf z}_{i,k}^t=h(\exp({\bm\zeta_{i,k}^\prime})\bar{\mathbf X}_i^{t|t-1},\bm 0),\, k=1,2,...,2L^\prime$
\State \quad \textcolor{blue}{// compute measurement covariance}
\State \quad $\bar{\mathbf z}_{i}^t=\sum_{k=0}^{2L^\prime}W_{i,k}^\prime\bar{\mathbf z}_{i,k}^t$
\State \quad ${\mathbf P_{zz,i}^t}=\sum_{k=0}^{2L^\prime}W_{i,k}^\prime({\mathbf z}_{i,k}^t-\bar{\mathbf z}_{i}^t)({\mathbf z}_{i,k}^t-\bar{\mathbf z}_{i}^t)^\top$
\State \quad ${\mathbf P_{xz, i}^t}=\sum_{k=0}^{2L^\prime}W_{i,k}^\prime(\bm{\alpha}_{i,0}^\prime-\bm{\alpha}_{i,k}^\prime)({\mathbf z}_{i,k}^t-\bar{\mathbf z}_{i}^t)^\top$
\State \quad \textcolor{blue}{// update the state and covariance}
\State \quad $\bar{\bm\xi}_i^t={\mathbf P_{xz,i}^t}{\mathbf P_{zz,i}^t}^{-1}(\mathbf z_i^t-\bar{\mathbf z}_i^t)$
\State \quad $\bar{\mathbf X}_i^t=\exp(\bar{\bm\xi}_i^t) \bar{\mathbf X}_i^{t|t-1}$
\State \quad ${\bar{\mathbf P}_{i}^t}={\bar{\mathbf P}_{i}^{t|t-1}}-{\mathbf P_{xz,i}^t}({\mathbf P_{xz,i}^t}{{\mathbf P_{zz,i}^t}}^{-1})^\top$\label{eq_cov_up}
\State \textbf{Output:} $\bar{\mathbf X}_i^{t}$, $\bar{\mathbf P}_i^{t}$
\end{algorithmic}
\end{algorithm}

As with standard diffusion-based distributed Kalman Filters \cite{Liang2021, Hu2012}, the proposed DIUKF-ICI framework consists of three main steps: propagation, incremental updates, and diffusion updates. The propagation follows the standard invariant UKF approach (\textbf{cf.} step 1, Algorithm \ref{alg2}), as outlined in \cite{BMB2017}, where each agent performs prediction using its own information of the system dynamics. Since this process is not the focus of this paper, we will only detail the incremental and diffusion updates in the following sections.

\subsection{Incremental update}
In order to derive the proper form for the incremental update, we first introduce the pseudo measurement matrix $\mathbf H_i^t$ \cite{CHEN2021} defined by
\begin{align}
\mathbf H_i^t&=({\mathbf P}_{xz,i}^t)^\top(\bar {\mathbf P}_i^{t|t-1})^{-1}
\end{align}
By using the matrix inverse lemma, the covariance update (\textbf{cf.} Algorithm \ref{alg2}, line. \eqref{eq_cov_up}) can be rewritten as
\begin{align}
(\bar{\mathbf P}_i^t)^{-1}&=(\bar{\mathbf P}_i^{t|t-1})^{-1}+(\mathbf H_i^t)^\top (\mathbf R_i^t)^{-1}\mathbf H_i^t
\end{align}
where 
\begin{align}
\mathbf R_i^t&=\mathbf P_{zz,i}^t-\mathbf H_i^t \mathbf P_{xz,i}^t
\end{align}
Given the individual estimate $(\bar {\mathbf X}_i^{t}, \bar {\mathbf P}_i^{t})$ obtained from the local UKF of each agent after the propagation step (\textbf{cf.} step 1, Algorithm \ref{alg2}), each agent takes its own measurements and receives the quantities $(\mathbf H_{j}^t)^\top (\mathbf R_j^t)^{-1}\mathbf H_{j}^t$, $(\mathbf H_{j}^t)^\top (\mathbf R_j^t)^{-1}(\mathbf z_j^t-\bar{\mathbf z}_j^t)$ from its neighbors. Then it performs the incremental update to computes a local estimate, denoted as $(\check {\mathbf X}_i^t, \check {\mathbf P}_i^t)$, as follows
\begin{align}\label{eq_inp1}
(\check{\mathbf P}_i^t)^{-1}&=(\bar{\mathbf P}_i^{t})^{-1}+\sum_{j\in\mathcal{N}_i^t, j\neq i}(\mathbf H_{j}^t)^\top (\mathbf R_j^t)^{-1}\mathbf H_{j}^t\nonumber\\
\check{\mathbf X}_i^t&=\exp(\bar{\bm \xi}_i^t)
\end{align}
where terms $\mathbf H_i^t$, $\mathbf P_{zz,t}^t$, $\mathbf P_{xz,t}^t$ are computed following Algorithm \ref{alg2}, and
\begin{align}\label{eq_inp2}
\bar{\bm \xi}_i^t&= \bar{\mathbf P}_i^{t}\sum_{j\in\mathcal{N}_i^t,j\neq i}(\mathbf H_{j}^t)^\top (\mathbf R_j^t)(\mathbf z_j^t-h_j(\mathbf X_j^t))
\end{align}

\subsection{Diffusion update}
The objective of the diffusion update is to compute an improved estimate for each agent by fusing the local estimates of its neighbors once the incremental update is completed. Since each agent shares a common system model and uses the neighbors' information to update its individual estimate in the incremental update step, the local estimates among agents are mutually correlated, with the correlations being unknown. Various fusion rules have been proposed to address the challenge of unknown correlations for a consistent estimate, with covariance intersection (CI)\cite{CI1997} and inverse covariance intersection (ICI) \cite{NOACK2017} being the most commonly used methods in the literature. Compared to the overly conservative CI method, which results in larger covariance estimates, ICI reduces conservatism and yields tighter covariance bounds, providing more accurate uncertainty estimates. However, ICI is currently limited to vector spaces and cannot be applied to Lie groups, as no such work has been developed yet. To ensure the fusion performance in diffusion updates for local estimates represented on Lie groups, this paper proposes a novel fusion algorithm that extends ICI to Lie groups.

Before introducing the proposed fusion algorithm, we first present several useful definitions and lemmas that will be employed in the algorithm's derivation.
Given the local estimate $(\check {\mathbf X}_i^t, \check {\mathbf P}_i^t)$ from agent $i$ at timestep $t$, let ${\bm{\xi}}_i^t$ denote the local estimate error between the local estimate $\check {\mathbf X}_i^t$ and the true state ${\mathbf X}^t$, which can be represented as
\begin{align}\label{eq_err}
\text{exp}\left({\bm\xi}_i\right)&={\mathbf X}^t(\check {\mathbf X}_i^t)^{-1}={\mathbf X}^t(\check {\mathbf X}_j^t)^{-1}(\check {\mathbf X}_j^t)(\check{\mathbf X}_i^t)^{-1}\nonumber\\
&=\exp({\bm\xi}_j^t)\check {\mathbf X}_j^t(\check {\mathbf X}_i^t)^{-1}
\end{align}
where ${\bm\xi}_j^t$ is the error of $j$'s estimate defined in $se_2(3)$, and $\check {\mathbf X}_j^t(\check {\mathbf X}_i^t)^{-1}$ denotes the error between estimate $i$ and $j$.
\begin{assumption}\label{ass_1}
The error ${\bm\xi}_i$ is small, i.e., there is no substantial difference between the local estimate $\check {\mathbf X}_i^t$ and the true state ${\mathbf X}_i^t$.
\end{assumption}
\begin{lemma}\label{lem_1} Given assumption \ref{ass_1}, each $\log\left(\check {\mathbf X}_j^t(\check {\mathbf X}_i^t)^{-1}\right)$, $\text{for} \,j\in\mathcal{N}_i^t$, can be treated as an estimate of ${\bm\xi}_i^t$, and ${\bm\xi}_j^t$ is the corresponding estimate error with error covariance ${\check{\mathbf P}_j^t}$. 
\end{lemma}

\begin{proof}
Following equation \eqref{eq_err}, $\check {\mathbf X}_j^t(\check {\mathbf X}_i^t)^{-1}$ can be represented using the Baker-Campbell-Hausdorff formula \cite{SSR2017} given as
\begin{align}
&\check {\mathbf X}_j^t(\check {\mathbf X}_i^t)^{-1}=\exp({\bm\xi}_i^t)\exp(-{\bm\xi}_j^t)\nonumber\\
&=\exp\left({\bm\xi}_i^t-{\bm\xi}_j^t+\frac{1}{2}[{\bm\xi}_i^t,{-\bm\xi}_j^t]+\frac{1}{12}[{\bm\xi}_i^t,{-\bm\xi}_j^t]+...\right)
\end{align}
Given assumption \ref{ass_1}, we can approximate $\check {\mathbf X}_i^t(\check {\mathbf X}_j^t)^{-1}$ by ignoring the high-order term, yielding
\begin{align}
\log\left(\check {\mathbf X}_j^t(\check {\mathbf X}_i^t)^{-1}\right)&\approx{\bm\xi}_i^t-{\bm\xi}_j^t
\end{align}
which can be equivalently written as:
\begin{align}
{\bm\xi}_j^t&\approx{\bm\xi}_i^t-\log\left(\check {\mathbf X}_j^t(\check {\mathbf X}_i^t)^{-1}\right)
\end{align}
which indicates that each $\log\left(\check {\mathbf X}_j^t(\check {\mathbf X}_i^t)^{-1}\right)$ can be treated as an estimate of ${\bm\xi}_i^t$ with error ${\bm\xi}_j^t$ and covariance $\check{\mathbf P}_j^t$. This completes the proof of Lemma \ref{lem_1}.
\end{proof}

Note that the local estimate error ${\bm\xi}_i^t$ can never be precisely known since the true state $\mathbf X_i^t$ is unknown. Nevertheless, 
we can still utilize the term $\log\left(\check {\mathbf X}_j^t(\check {\mathbf X}_i^t)^{-1}\right)$ to estimate the error ${\bm\xi}_i^t$ as stated in Lemma \ref{lem_1}. 
Based on this insight, we aim to fuse all the estimate pairs $\left(\log\left(\check {\mathbf X}_j^t(\check {\mathbf X}_i^t)^{-1}\right),\,\check {\mathbf P}_j^t\right)$ to obtain an improved estimate for the error ${\bm\xi}_i^t$, denoted as $\left(\hat{\bm\xi}_i^t, \hat {\mathbf P}_i^t\right)$. Specifically, we employ the multi-fusion ICI algorithm \cite{AJ2020} to fuse all the pairs $\left(\log\left(\check {\mathbf X}_j^t(\check {\mathbf X}_i^t)^{-1}\right),\,\check {\mathbf P}_j^t\right)$. The fused covariance $\hat {\mathbf P}_i^t$ can be computed as
\begin{align}\label{eq_fus_cov}
\hat{\mathbf P}_i^t&=\begin{bmatrix}\sum_{j\in\mathcal{N}_i^t}(\check{\mathbf P}_j^t)^{-1}-(\left|\mathcal{N}_i^t\right|-1)(\mathbf P_{\Gamma_i}^t)^{-1}\end{bmatrix}^{-1}   
\end{align}
where 
\begin{align}
\mathbf P_{\Gamma_i}^t&=\sum_{j\in\mathcal{N}_i^t}\omega_{ij}^t \check{\mathbf P}_j^t
, \, \text{with}\, \omega_{ij}^t\in[0,1],\,\sum_{j\in\mathcal{N}_i^t}\omega_{ij}=1
\end{align}
The weight $\omega_{ij}$ can be computed using the algorithm in \cite{AJ2020}, and the fused error $\hat{\bm\xi}_i^t$ is given by
\begin{align}
\hat{\bm\xi}_i^t&=\hat{\mathbf P}_i^t \bm\Gamma_i^t
\log\left(\check {\mathbf X}_j^t(\check {\mathbf X}_i^t)^{-1}\right)\\
\bm\Gamma_i^t&=
\sum_{j\in\mathcal{N}_i^t}\left((\check{\mathbf P}_j^t)^{-1}-(\left|\mathcal{N}_i^t\right|-1)\omega_{ij}^t(\mathbf P_{\Gamma_i}^t)^{-1}\right)    
\end{align}
Compared with the statement in Lemma \ref{lem_1} that approximates the local estimate error ${\bm\xi}_i^t$ by the term $\log\left(\check {\mathbf X}_j^t(\check {\mathbf X}_i^t)^{-1}\right)$, the fused error term $\hat{\bm\xi}_i^t$ can be treated as an improved estimate of the error ${\bm\xi}_i^t$. As a result, the fused $\hat{\bm\xi}_i^t$ can then be used to correct the local estimate $\check{\mathbf X}_i^t$ of each agent as
\begin{align}\label{eq_crr}
\hat{\mathbf X}_i^t&=\exp(\hat{\bm\xi}_i^t)\check{\mathbf X}_i^t    
\end{align}
resulting in a fused state estimate. To derive the fused covariance, we first define $\bm\varepsilon_i^t$ as the estimate error of the fused state $\hat{\mathbf X}_i^t$ given by
\begin{align}
\exp(\bm\varepsilon_i^t)&={\mathbf X}^t (\hat{\mathbf X}_i^t)^{-1}=\exp({\bm\xi}_i^t)\check{\mathbf X}_i^t (\hat{\mathbf X}_i^t)^{-1}\nonumber\\
&=\exp({\bm\xi}_i^t)\exp(-\hat{\bm\xi}_i^t)
\end{align}
By using the BCH formula \cite{SSR2017} and ignoring the high-order terms, the above equation yields
\begin{align}
\bm\varepsilon_i^t\approx{\bm\xi}_i^t-\hat{\bm\xi}_i^t
\end{align}
where $\hat{\mathbf P}_i^t$ denotes the covariance. Consequently, $\hat{\mathbf P}_i^t$ can be directly treated as estimate covariance for $\hat{\mathbf X}_i^t$.

\section{Performance Analysis}


\begin{lemma}\label{lem_2} \cite{CHEN2021} Assume that $A^t,B^t\in \mathbb R^{n\times n}$, for $t=0,1,\cdots$, are positive define invertible matrices sequence, and $A^t,B^t$ are bounded by $\underline a \mathbf I\leq A^t\leq \bar a \mathbf I$, $\underline b \mathbf I\leq B^t\leq \bar b \mathbf I$, then the following inequality holds
\begin{align}
\mathbf 0\leq(A^t+B^t)^{-1}\leq (A^t)^{-1}
\end{align}
\end{lemma}

\begin{assumption}\label{ass_2}
Assume that the individual estimate resulting from the local UKF is bounded, i.e., $(\underline \gamma \mathbf I\leq\bar{\mathbf P}_i^t  \leq\bar \gamma\mathbf{I})$ for all \(i = 1, 2, \dots, N\), where $\underline\gamma, \bar\gamma$ are positive real numbers.    
\end{assumption}

\begin{assumption}\label{ass_3}
There exist positive real numbers $\underline r, \bar r$ such that $\underline r \mathbf I\leq\mathbf R_i^t\leq \bar r\mathbf I$.
\end{assumption}

\begin{remark}
Assumption \ref{ass_2} and \ref{ass_3} are commonly adopted in the literature \cite{CHEN2021, Hu2012} for analyzing the performance of the DSE system.
\end{remark}

\begin{lemma}\label{lem_3}\cite{HR2012}
Let $\bar{P}$ be the weighted sum of $N$ positive definite matrices $ P_i \in \mathbb{R}^{n \times n} $, for all $i = 1, 2, \cdots, N $, i.e.,$ \bar{P} = \sum_{i=1}^N \omega_i P_i $, with $ \omega_i \in [0, 1] $ and $ \sum_{i=1}^N \omega_i = 1 $, and matrix $P_i$ is bounded by $ \underline \alpha \mathbf I\leq P_i \leq \bar\alpha \mathbf{I} $, where $\underline\alpha, \bar\alpha$ are positive real numbers. Then matrix $\bar{P}$ is bounded by $\underline\alpha \mathbf I\leq\bar P\leq \bar\alpha \mathbf I$, and the following inequality holds
\begin{align}
(\bar{P})^{-1} = (\sum_{i=1}^N \omega_i P_i)^{-1}\leq \sum_{i=1}^N \omega_i (P_i)^{-1}
\end{align}
\end{lemma}

\begin{theorem}\label{Th1}
Consider the nonlinear system \eqref{eq_sys} under Assumption \ref{ass_2} and \ref{ass_3}.
The estimated covariance $\hat{\mathbf P}_i^t$ for each agent, obtained through Algorithm \ref{alg1}, is uniformly bounded.
\end{theorem}

\begin{proof}
Consider the incremental update for each agent under Assumption \ref{ass_2} and \ref{ass_3}. According to Lemma \ref{lem_2}, we have
\begin{align}
(\check{\mathbf P}_i^t)&=\left((\bar{\mathbf P}_i^{t})^{-1}+\sum_{j\in\mathcal{N}_i^t,j\neq i}(\mathbf H_{j}^t)^\top (\mathbf R_j^t)^{-1}\mathbf H_{j}^t\right)^{-1}\nonumber\\
&\leq\bar{\mathbf P}_i^{t}\leq \bar\gamma \mathbf I
\end{align}
This shows the covariance $\check{\mathbf P}_i^t$ generated from the incremental update, is bounded for all agents. Building on this result, we further derive that the fused covariance after the diffusion update, is also bounded. In particular, from \eqref{eq_fus_cov}, the following inequality holds according to Lemma \ref{lem_3}
\begin{align}\label{eq_cov_bd}
(\hat {\mathbf P}_i^t)^{-1}&=\sum_{j\in\mathcal{N}_i^t}(\check{\mathbf P}_j^t)^{-1}-(\left|\mathcal{N}_i^t\right|-1)(\mathbf P_{\Gamma_i}^t)^{-1}\nonumber\\
&\geq\sum_{j\in\mathcal{N}_i^t}(\check{\mathbf P}_j^t)^{-1}-(\left|\mathcal{N}_i^t\right|-1)\sum_{j\in\mathcal{N}_i^t}\omega_{ij}^t (\check{\mathbf P}_j^t)^{-1}\nonumber\\
&\geq \sum_{j\in\mathcal{N}_i^t}\omega_{ij}^t (\check{\mathbf P}_j^t)^{-1} \geq \frac{1}{\bar\gamma} \mathbf I
\end{align}
since matrix $\hat{\mathbf P}_i^t$ is the fused covariance matrix which is positive-define \cite{AJ2020}, thus \eqref{eq_cov_bd} implies that
\begin{align}
\hat {\mathbf P}_i^t\leq \bar\gamma \mathbf I
\end{align}
which completes the proof of Theorem \ref{Th1}.

\end{proof}

\section{Numerical Studies}
In this section, we will demonstrate the performance of the proposed DIUKF-ICI algorithm for target tracking using extensive Monte-Carlo simulations. The scenario involves tracking the 3D motion of a quadrotor with eight Ultra Wide Band (UWB) sensors. Each sensor has a limited sensing range of 5 meters, which means the sensors cannot detect the quadrotor if it moves outside this range. Additionally, we introduce the concept of \textit{communication rate}, expressed as a percentage, to quantify the connectivity of the sensor network. For instance, a $50\%$ communication rate implies that the probability of successfully transmitting an information packet between two nodes is $50\%$. 

\textbf{Simulation models and parameters:}
The target state in the simulation is represented by a $SE(3)$ group as follows
\begin{align}
    \mathbf{X} = \begin{bmatrix}
    {{^{G}_L\mathbf{R}}} & {^G{\mathbf v}} & {^G{\mathbf p}}  \\
    \mathbf{0}_{1 \times 2} & 1 & 0  \\
    \mathbf{0}_{1 \times 2} & 0 & 1 \\
    \end{bmatrix} \in \mathbb{R}^{5 \times 5}
\end{align}
where ${^{G}_L\mathbf{R}} \in SO(3)$ denotes the target's orientation from the target local frame $L$ to the global frame $G$. ${^G{\mathbf v}}, {^G{\mathbf p}} \in \mathbb{R}^3$ represents the target's velocity and position. 

The proposed algorithm requires the target’s dynamics model for state propagation. One approach to predict the target's state—encompassing position, orientation, and velocity—is through the use of inertial measurement unit (IMU) data for dynamics modeling. In this paper, we directly utilize IMU propagation model as the motion model to generate a predicted state estimate \cite{BM2018}. The frequency of the IMU and UWB measurements are set to be $100 Hz$ and $10 Hz$, respectively. The
IMU's linear acceleration and angular velocity are assumed to progress as random walk driven by Gaussian, which are selected to be $[0.02, 0.02, 0.05]^\top {m/s^2/\sqrt{Hz}}$, and $[0.005, 0.005, 0.005]^\top {deg/s/\sqrt{Hz}}$.
For the measurement model, we apply the UWB ranging measurement model to the sensor networks as
\begin{align}
z_i^t = \lVert{^G\mathbf{p}^t - ^G\mathbf{p}_a}\rVert + \mathbf n_u^t
\end{align}
where ${^G\mathbf{p}_a}$ denotes the UWB anchor positions, and $\mathbf n_u^t$ is assumed to be Gaussian with $\sigma=0.10 m$. The anchor positions in our simulation are shown in Table \ref{tab:sensor_pos}.

\begin{table}[h]
    \centering
    \caption{Coordinates of the 8 Sensors}
    \vspace{-0.5ex}
    \begin{tabular}{cccc}
        \toprule
        Sensor & Coordinates (m) & Sensor & Coordinates (m) \\
        \midrule
        1 & [-5 -5 0] & 5 & [-5 -5 5] \\
        2 & [15 -5 0] & 6 & [15 -5 5] \\
        3 & [-5 15 0] & 7 & [-5 15 5] \\
        4 & [15 15 0] & 8 & [15 15 5] \\
        \bottomrule
    \end{tabular}
    \label{tab:sensor_pos}
\end{table}

\textbf{Simulation results:}
As shown in Figure \ref{fig_traj}, the drone (target) follows two pre-designed trajectories and we can see if the proposed algorithm can accurately track them. The initial pose of each trajectory is shown in Table \ref{tab:pos_init}.
\begin{table}[h]
    \centering
    \caption{Initialization of Trial}
    \begin{tabular}{ccc}
        \toprule
        Traj. & Position (m) & Orientation (rad) \\
        \midrule
        1 & $\begin{bmatrix} -2.0000 & 2.0000 & 0.8000 \end{bmatrix}^\top$ & $\begin{bmatrix} 0 & -\pi & 0 \end{bmatrix}^\top$ \\
        2 & $\begin{bmatrix} -2.0000 & 0.0000 & 0.5000 \end{bmatrix}^\top$ & $\begin{bmatrix} 0 & -\pi & 0 \end{bmatrix}^\top$ \\
        \bottomrule
    \end{tabular}
    \label{tab:pos_init}
\end{table}

We conduct 50 trails of Monte-Carlo simulations, and plot the estimation results by node 1 as a representative result, shown in Figure \ref{fig_traj}.  It can be seen that the proposed DIUKF-ICI can achieve satisfying tracking performance. The estimation results are also quantified by rooted mean square error (RMSE) to evaluate the accuracy. Figure \ref{fig_rmse} shows the averaged position RMSE (PRMSE) and orientation RMSE (ORMSE) across different percentages of the communication as well as the centralized scenarios. It can be observed that the proposed algorithm is robust to time-varying communication topology, demonstrating superior performance in the results of both position and orientation estimations.


\renewcommand{\arraystretch}{1.5}
\begin{table*}[h!]
    \centering
    \caption{Performance analysis under various communication rates}
    \begin{tabular}{cccccccccc}
        \toprule
        \multicolumn{2}{c}{\textbf{Communication Rate}} & \multicolumn{2}{c}{10\%} & \multicolumn{2}{c}{40\%} & \multicolumn{2}{c}{70\%} & \multicolumn{2}{c}{100\%} \\
        \cmidrule(lr){1-2} \cmidrule(lr){3-4} \cmidrule(lr){5-6} \cmidrule(lr){7-8} \cmidrule(lr){9-10}
        Algorithm & Trajectory & PRMSE & ORMSE & PRMSE & ORMSE & PRMSE & ORMSE & PRMSE & ORMSE \\
        \midrule
        DIUKF-ICI & 1 & 0.0647 & 1.8731 & 0.0496 & 0.8201 & 0.0477 & 0.5528 & 0.0499 & 0.5222 \\
        DIUKF-ICI & 2 & 0.0669 & 1.6890 & 0.0519 & 0.9236 & 0.0501 & 0.6495 & 0.0494 & 0.6172 \\
        DIEKF-ICI & 1 & 0.0696 & 1.9134 & 0.0538 & 0.8729 & 0.0479 & 0.8066 & 0.0526 & 0.7087 \\
        DIEKF-ICI & 2 & 0.0693 & 1.9916 & 0.0517 & 0.9847 & 0.0574 & 0.8723 & 0.0555 & 0.9225 \\
        DIEKF-CI & 1 & 0.0737 & 1.9949 & 0.0526 & 0.9082 & 0.0581 & 0.9529 & 0.0558 & 0.8729 \\
        DIEKF-CI & 2 & 0.0735 & 1.8970 & 0.0536 & 0.9716 & 0.0579 & 0.9759 & 0.0562 & 0.9379 \\
        \bottomrule
    \end{tabular}\label{tab:error_analysis}
\end{table*}

\begin{figure}[h]
    \centering
    \begin{subfigure}[b]{0.23\textwidth}
        \centering
        \includegraphics[width=\textwidth]{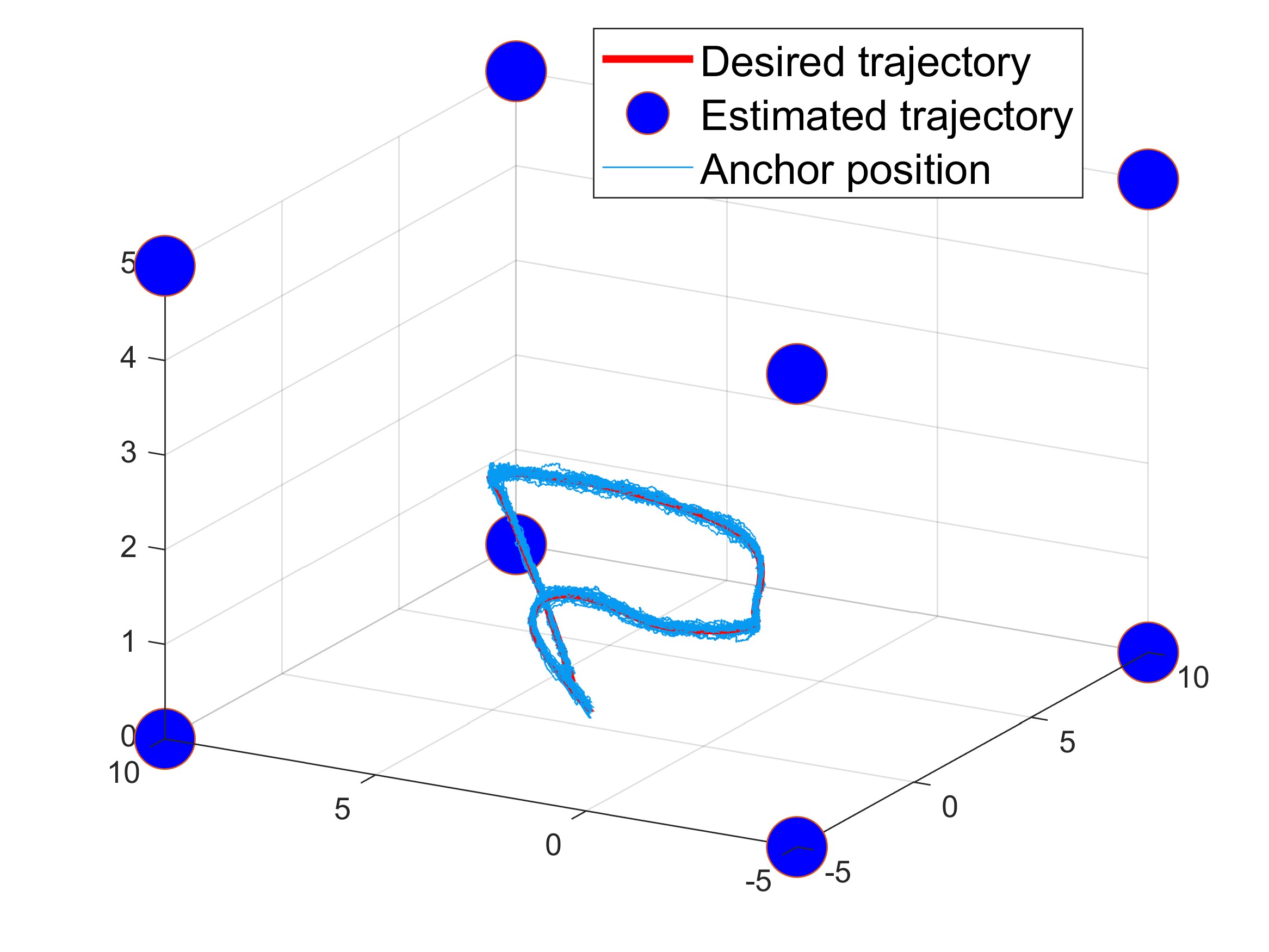}  
        \caption{Trajectory 1}
        \label{fig:image1}
    \end{subfigure}
    \hfill
    \begin{subfigure}[b]{0.23\textwidth}
        \centering
        \includegraphics[width=\textwidth]{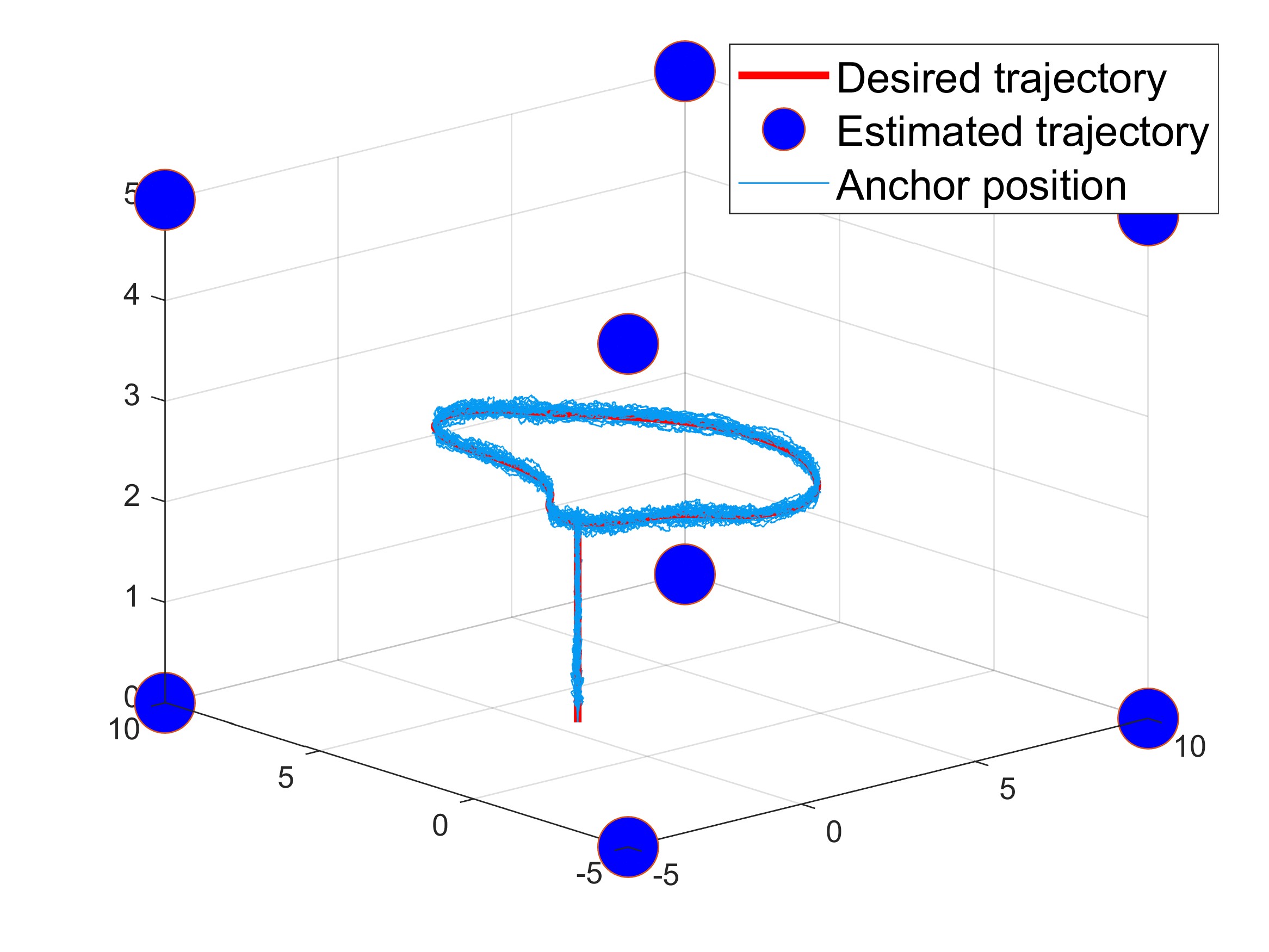}  
        \caption{Trajectory 2}
        \label{fig:image2}
    \end{subfigure}
    \caption{Estimation results of DIUKF-ICI on different trajectories at 70\% communication rate}
    \label{fig_traj}
\end{figure}


\begin{figure}[h]
    \centering
    \begin{subfigure}[b]{0.5\textwidth}
        \centering
        \includegraphics[width=\textwidth]{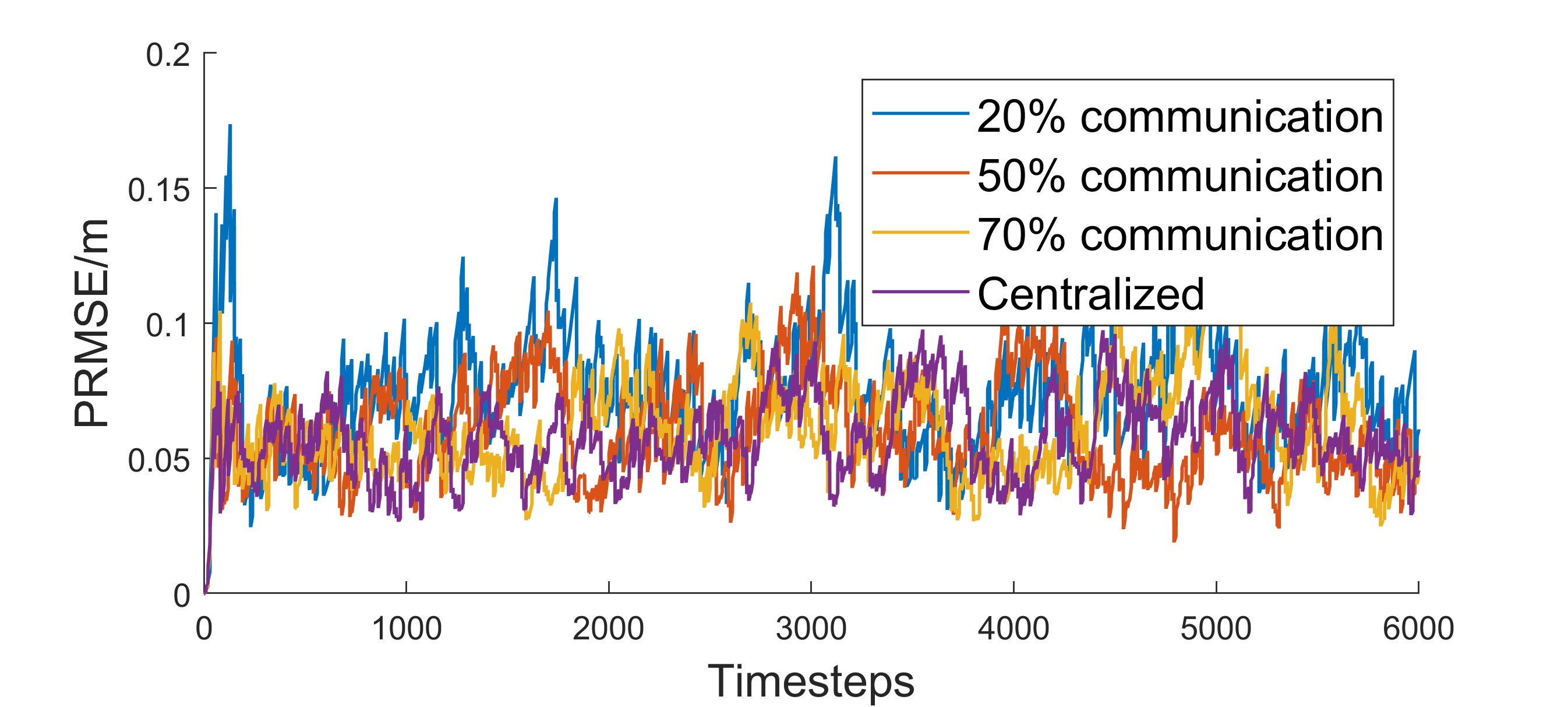}
        \caption{PRMSE}
        \label{fig:subfig1}
    \end{subfigure}
    \hfill
    \begin{subfigure}[b]{0.5\textwidth}
        \centering
        \includegraphics[width=\textwidth]{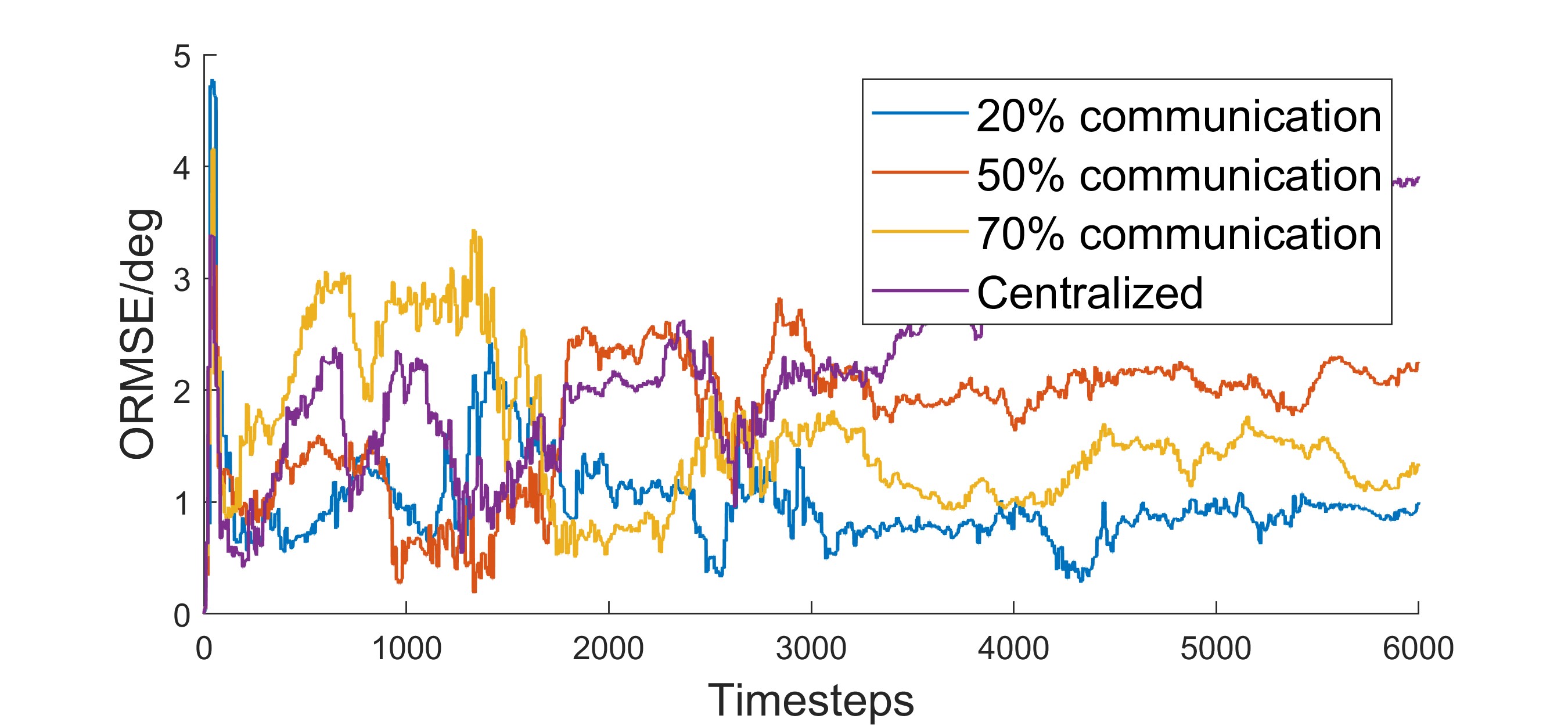}
        \caption{ORMSE}
        \label{fig:subfig2}
    \end{subfigure}
    \caption{PRMSE and ORMSE of trajectory 1 under different communications}
    \label{fig_rmse}
\end{figure}
To further demonstrate the effectiveness of the proposed algorithm and the broad applicability of the proposed extended ICI fusion rule on Lie groups, we also present the results of the distributed IEKF based on CI (DIEKF-CI) as a comparison, and incorporate the extended ICI fusion rules into the distributed EKF-based frameworks. From Table \ref{tab:error_analysis}, it is clear that the DIUKF-ICI outperforms the DIEKF-CI in estimating both position and orientation, as the UKF avoids the need for linearization, making it more suitable for handling nonlinearities inherent in the system. 
In addition, the proposed extended ICI fusion rule also shows strong performance in the distributed IEKF-based frameworks, which demonstrates the generalization of this method, as it achieves favorable results in both UKF-based and EKF-based frameworks.



\section{Conclusion}
This paper presents a novel distributed invariant Unscented Kalman Filter (DIUKF) using inverse covariance intersection (ICI) for sensor networks. The proposed algorithm extends the distributed UKF framework to Lie groups, enabling local estimates
to be fused with intermediate information from neighboring agents on Lie groups. In addition, the proposed algorithm is fully distributed which requires only the local information.
The performance and effectiveness of the proposed algorithm is validated through simulations in distributed 3-D target tracking scenarios. In the future works, we plan to implement the proposed algorithms in real-world robot platforms to test its performance.

\section{acknowledgement}
The authors would like to thank Dr. Xuan Wang for his useful suggestions and meaningful discussions.

\bibliographystyle{IEEEtran}
\bibliography{main}

\begin{thebibliography}{10}
\providecommand{\url}[1]{#1}
\csname url@samestyle\endcsname
\providecommand{\newblock}{\relax}
\providecommand{\bibinfo}[2]{#2}
\providecommand{\BIBentrySTDinterwordspacing}{\spaceskip=0pt\relax}
\providecommand{\BIBentryALTinterwordstretchfactor}{4}
\providecommand{\BIBentryALTinterwordspacing}{\spaceskip=\fontdimen2\font plus
\BIBentryALTinterwordstretchfactor\fontdimen3\font minus \fontdimen4\font\relax}
\providecommand{\BIBforeignlanguage}[2]{{%
\expandafter\ifx\csname l@#1\endcsname\relax
\typeout{** WARNING: IEEEtran.bst: No hyphenation pattern has been}%
\typeout{** loaded for the language `#1'. Using the pattern for}%
\typeout{** the default language instead.}%
\else
\language=\csname l@#1\endcsname
\fi
#2}}
\providecommand{\BIBdecl}{\relax}
\BIBdecl

\bibitem{KUMAR2021109558}
\BIBentryALTinterwordspacing
M.~Kumar and S.~Mondal, ``Recent developments on target tracking problems: A review,'' \emph{Ocean Engineering}, vol. 236, p. 109558, 2021. [Online]. Available: \url{https://www.sciencedirect.com/science/article/pii/S0029801821009471}
\BIBentrySTDinterwordspacing

\bibitem{NOACK2017}
\BIBentryALTinterwordspacing
B.~Noack, J.~Sijs, M.~Reinhardt, and U.~D. Hanebeck, ``Decentralized data fusion with inverse covariance intersection,'' \emph{Automatica}, vol.~79, pp. 35--41, 2017. [Online]. Available: \url{https://www.sciencedirect.com/science/article/pii/S0005109817300298}
\BIBentrySTDinterwordspacing

\bibitem{Hu2012}
J.~Hu, L.~Xie, and C.~Zhang, ``Diffusion kalman filtering based on covariance intersection,'' \emph{IEEE Transactions on Signal Processing}, vol.~60, no.~2, pp. 891--902, 2012.

\bibitem{Ren2007}
W.~Ren, R.~W. Beard, and E.~M. Atkins, ``Information consensus in multivehicle cooperative control,'' \emph{IEEE Control Systems Magazine}, vol.~27, no.~2, pp. 71--82, 2007.

\bibitem{CHEN2021}
\BIBentryALTinterwordspacing
H.~Chen, J.~Wang, C.~Wang, J.~Shan, and M.~Xin, ``Distributed diffusion unscented kalman filtering based on covariance intersection with intermittent measurements,'' \emph{Automatica}, vol. 132, p. 109769, 2021. [Online]. Available: \url{https://www.sciencedirect.com/science/article/pii/S0005109821002892}
\BIBentrySTDinterwordspacing

\bibitem{Sun2023}
T.~Sun and M.~Xin, ``Inverse-covariance-intersection-based distributed estimation and application in wireless sensor network,'' \emph{IEEE Transactions on Industrial Informatics}, vol.~19, no.~10, pp. 10\,079--10\,090, 2023.

\bibitem{ZYICI}
Y.~Zhou, Y.~Liu, and X.~Wang, ``Distributed estimation for a 3-d moving target in quaternion space with unknown correlation,'' in \emph{2024 IEEE Conference on Control Technology and Applications (CCTA)}, 2024, pp. 394--399.

\bibitem{Jie2023}
J.~Xu, P.~Zhu, Y.~Zhou, and W.~Ren, ``Distributed invariant extended kalman filter using lie groups: Algorithm and experiments,'' \emph{IEEE Transactions on Control Systems Technology}, vol.~31, no.~6, pp. 2777--2789, 2023.

\bibitem{ZH2024}
\BIBentryALTinterwordspacing
C.~Zhang, J.~Qin, C.~Yan, Y.~Shi, Y.~Wang, and M.~Li, ``Towards invariant extended kalman filter-based resilient distributed state estimation for moving robots over mobile sensor networks under deception attacks,'' \emph{Automatica}, vol. 159, p. 111408, 2024. [Online]. Available: \url{https://www.sciencedirect.com/science/article/pii/S0005109823005757}
\BIBentrySTDinterwordspacing

\bibitem{Liang2021}
L.~Li and M.~Yang, ``Joint localization based on split covariance intersection on the lie group,'' \emph{IEEE Transactions on Robotics}, vol.~37, no.~5, pp. 1508--1524, 2021.

\bibitem{LE2023}
\BIBentryALTinterwordspacing
J.-G. Lee, Q.~V. Tran, K.-H. Oh, P.-G. Park, and H.-S. Ahn, ``Distributed object pose estimation over strongly connected networks,'' \emph{Systems \& Control Letters}, vol. 175, p. 105505, 2023. [Online]. Available: \url{https://www.sciencedirect.com/science/article/pii/S016769112300052X}
\BIBentrySTDinterwordspacing

\bibitem{Zarei2024}
M.~Zarei and R.~Chhabra, ``Consistent fusion of correlated pose estimates on matrix lie groups,'' \emph{IEEE Robotics and Automation Letters}, vol.~9, no.~7, pp. 6584--6591, 2024.

\bibitem{BAB2017}
A.~Barrau and S.~Bonnabel, ``The invariant extended kalman filter as a stable observer,'' \emph{IEEE Transactions on Automatic Control}, vol.~62, no.~4, pp. 1797--1812, 2017.

\bibitem{BMB2017}
M.~Brossard, S.~Bonnabel, and J.-P. Condomines, ``Unscented kalman filtering on lie groups,'' in \emph{2017 IEEE/RSJ International Conference on Intelligent Robots and Systems (IROS)}, 2017, pp. 2485--2491.

\bibitem{CI1997}
S.~Julier and J.~Uhlmann, ``A non-divergent estimation algorithm in the presence of unknown correlations,'' in \emph{Proceedings of the 1997 American Control Conference (Cat. No.97CH36041)}, vol.~4, 1997, pp. 2369--2373 vol.4.

\bibitem{SSR2017}
T.~D. Barfoot, \emph{State Estimation for Robotics}.\hskip 1em plus 0.5em minus 0.4em\relax Cambridge University Press, 2017.

\bibitem{AJ2020}
J.~Ajgl and O.~Straka, ``Inverse covariance intersection fusion of multiple estimates,'' in \emph{2020 IEEE 23rd International Conference on Information Fusion (FUSION)}, 2020, pp. 1--8.

\bibitem{HR2012}
R.~A. Horn and C.~R. Johnson, \emph{Matrix Analysis}.\hskip 1em plus 0.5em minus 0.4em\relax Cambridge University Press, 2012.

\bibitem{BM2018}
M.~Brossard, S.~Bonnabel, and A.~Barrau, ``Invariant kalman filtering for visual inertial slam,'' in \emph{2018 21st International Conference on Information Fusion (FUSION)}, 2018, pp. 2021--2028.

\end{thebibliography}

\end{document}